%% file: main.tex
\documentclass[conference]{IEEEtran}
\IEEEoverridecommandlockouts
\usepackage{cite}
\usepackage{amsmath,amssymb,amsfonts}
\usepackage{algorithmic}
\usepackage{graphicx}
\usepackage{textcomp}
\usepackage{amsthm}
\usepackage{tikz}

\newcommand{\SNR}{\mathrm{SNR}}

\newcommand{\rmin}{r_{\mathrm{min}}}
\newcommand{\rmax}{r_{\mathrm{max}}}
\newcommand{\rEarth}{r_{\oplus}}
\newcommand{\phiuser}{\phi_{\mathrm{u}}}
\newcommand{\phisat}{\phi_{\mathrm{s}}}
\newcommand{\Phisat}{\Phi_{\mathrm{s}}}

\newcommand{\thmin}{\theta_{\mathrm{min}}}

\newcommand{\Prob}{{\mathbb{P}}}

\newcommand{\Pc}{P_\mathrm{c}}
\newcommand{\PV}{P_\mathrm{V}}

\newcommand{\ps}{p_\mathrm{s}}

\newcommand{\phiu}{\phi_\mathrm{u}}

\newcommand{\Neff}{N_\mathrm{eff}}
\newcommand{\Nact}{N_\mathrm{act}}
\newcommand{\deltaeff}{\delta_\mathrm{eff}}
\newcommand{\deltaact}{\delta_\mathrm{act}}
\newtheorem{proposition}{Proposition}

\newtheorem{lem}{Lemma}
\newtheorem{cor}{Corollary}

\IEEEoverridecommandlockouts
\IEEEpubid{\makebox[\columnwidth]{978-1-7281-4490-0/20/\$31.00 \copyright~2020~IEEE}
\hspace{\columnsep}\makebox[\columnwidth]{ }}

\begin{document}
\setlength{\parskip}{0pt}
\title{Stochastic Analysis of Satellite Broadband\\by Mega-Constellations with Inclined LEOs}
\author{Niloofar~Okati and Taneli~Riihonen\\
Faculty of Information Technology and Communication Sciences,
Tampere University, Finland\\
e-mail: \{{\tt niloofar.okati}, {\tt taneli.riihonen}\}{\tt @tuni.fi}}

\maketitle

\let\thefootnote\relax\footnotetext{This research work was supported by a Nokia University Donation.}

\begin{abstract}
As emerging massive constellations are intended to provide seamless connectivity for remote areas using hundreds of small low Earth orbit (LEO) satellites, new methodologies have great importance to study the performance of these networks. In this paper, we derive both downlink and uplink analytical expressions for coverage probability and data rate of an inclined LEO constellation under general fading, regardless of exact satellites' positions. Our solution involves two phases as we, first, abstract the network into a uniformly distributed network. Secondly, we obtain a new parameter,  effective number of satellites, for every user's latitude which compensates for the performance mismatch between the actual and uniform constellations. In addition to exact derivation of the network performance metrics, this study provides insight into selecting the constellation parameters, e.g., the total number of satellites, altitude, and inclination angle.
\end{abstract}


\section{Introduction}
\label{sec:intro}

A constellation of low Earth orbit (LEO) satellites can provide infrastructure for ubiquitous connectivity with low round-trip delay---compared to geostationary satellites---when terrestrial networks are not available or economically reasonable to deploy \cite{46,47}. Technological advancements along with the need for seamless connectivity have emerged the utilization of massive satellite networks and, consequently, the research around this topic.  

The uplink outage probability in the presence of interference was evaluated for two LEO constellations through time-domain simulations in \cite{1}.
A performance study of Iridium constellation was presented in \cite{4} in terms of the distribution of the number of handoffs involved in a
single transaction duration and the average call drop probability. The effect of traffic non-uniformity was studied in \cite{38} by assuming hexagonal service areas for satellites.

A general expression for a single LEO satellite's visibility time was provided in \cite{8}, but it is incapable of concluding the general distribution of visibility periods for any arbitrarily positioned user. The deterministic model in \cite{8} was then developed by a statistical analysis of coverage time in mobile LEO during a satellite visit \cite{49}.
{In \cite{Afshang}, the Doppler shift magnitude of a LEO network is characterized for a single spotbeam by using tools from stochastic geometry.} Resource control of a hybrid satellite--terrestrial network was performed in \cite{51} with two objectives of maximizing the delay-limited capacity
and minimizing the outage probability.
A hybrid satellite--terrestrial network to assist 5G infrastructure has been analyzed by considering only one spotbeam~\cite{50,52}.

In the current literature around communication satellites' performance, the network analysis is limited to deterministic simulation-based studies, simplifying the network by considering specific constellations with a limited number of satellites, and assuming specific coverage footprints for satellites. Therefore, a comprehensive method that fits any constellation with arbitrary parameters is missing from the scientific literature. {In our recent study \cite{okati2020downlink}, downlink performance of a massive LEO constellation was investigated by assuming uniform distribution for satellites. However, the performance mismatch between actual and uniform constellations was compensated only through numerical mean absolute error minimization.}

In this paper, we provide a mathematical framework for downlink and uplink coverage probability and data rate analysis of an inclined LEO constellation under a general fading model. For our derivations, first, we assume the satellites are distributed uniformly on the orbital shell. Later, the mismatch between the actual and uniform constellations is compensated by deriving a new parameter as the effective number of satellites. Finally, the mathematical expressions are verified through simulations and the main findings of this paper are demonstrated for different network parameters, e.g., the total number of satellites, altitude and minimum elevation angle required for a satellite to be visible to the user. 
The results obtained in this paper are scalable for numerous problems in massive satellite networks.
 
The organization of the remainder of this paper is as follows. Section~\ref{sec:sysmod} describes the system model for an inclined LEO constellation. As for the main results, in Sections~\ref{sec:performance}, we derive analytical expressions for coverage probability and average achievable data rate for a terrestrial user and introduce the concept of effective number of satellites. Numerical results are provided in Section~\ref{sec:Numerical Results} for studying the effect of key system parameters. Finally, we conclude the paper in Section~\ref{sec:Conclusion}.

\section{System Model}
\label{sec:sysmod}

Let us consider a LEO communication satellite constellation, as shown in Fig.~\ref{fig:system_model}, that consists of $\Nact$ satellites, which are placed on low circular orbits with the same inclination angle and altitude denoted by $\iota$ and $\rmin$, respectively. The altitude parameter $\rmin$ has the subscript because it specifies also the minimum possible distance between a satellite and a user on Earth (that is realized when it is at the zenith).

User terminals are located on the surface of Earth that is approximated as a perfect sphere. We assume that wireless transmissions propagate to/from a user from/to all and only the satellites that are elevated above the horizon to an angle of $\theta_\mathrm{s}\geq \thmin$. Correspondingly, $\rmax$ denotes the maximum possible distance at which a satellite and a user may be able to communicate (that is realized when $\theta_\mathrm{s}=\thmin$), and
\begin{align}
\frac{\rmax}{\rEarth} = \sqrt{\frac{\rmin}{\rEarth}\left(\frac{\rmin}{\rEarth}+2\right) + \sin^2(\thmin)}-\sin(\thmin),
\end{align}
where $\rEarth \approx 6371$~km denotes Earth's radius. Conversely, the latitudes, where a terrestrial user may be able to establish connection with any satellite at all, are limited by
\begin{align}
\label{phi_max}
|\phiuser| \leq \iota + \cos^{-1}\left(\frac{\rEarth^2+\rEarth\rmin+\left(\rmin^2-\rmax^2\right)/2}{\rEarth(\rEarth+\rmin)}\right).
\end{align}
For instance, with satellite altitudes of $\rmin=500$~km and $\rmin=2000$~km, global coverage up to poles for $\thmin=0^{\circ}$ is possible only if $\iota > 68^{\circ}$ and $\iota > 49^{\circ}$, respectively. Using (\ref{phi_max}), the minimum altitude which provides global coverage is given as
\begin{align}
\rmin \geq \frac{\rEarth\cos{(\thmin)}}{\sin{(\iota-\thmin)}}-\rEarth.
\end{align}

Each user is associated with the nearest satellite that is referred to as the serving satellite in what follows. We assume that co-channel interference mitigation has been implemented properly so that the network performance becomes noise-limited. The distances from the user to the serving satellite and the other satellites are denoted by $r_0$ and $r_n$, \mbox{$n=1,2,\ldots,\Nact-1$}, respectively, while $G_0$ and $G_n$ represent the corresponding channel gains. Obviously, $G_n=0$ if $r_n > \rmax$ for some $n=0,1,\ldots,\Nact-1$.

Based on the above modeling, the signal-to-noise ratio (SNR) at the receiver can be expressed as
\begin{align}
\label{eq:SNR}
\SNR = \left\{
\begin{array}{ll}
\displaystyle
\frac{\ps G_0 r_0^{-\alpha}}{\sigma^2}, & r_0\leq\rmax,\\
0, & \text{otherwise},\\
\end{array} \right.
\end{align}
where we assume that the user's receiver is subject to additive white Gaussian noise with constant power $\sigma^2$, and the parameter $\alpha$ is a path loss exponent.

\begin{figure}[t]
\begin{centering}
\input{system_model.tex}
\end{centering}
\caption{A sketch of the considered system model, where satellites are distributed uniformly over the inclined orbits.}
\label{fig:system_model} 
\end{figure}
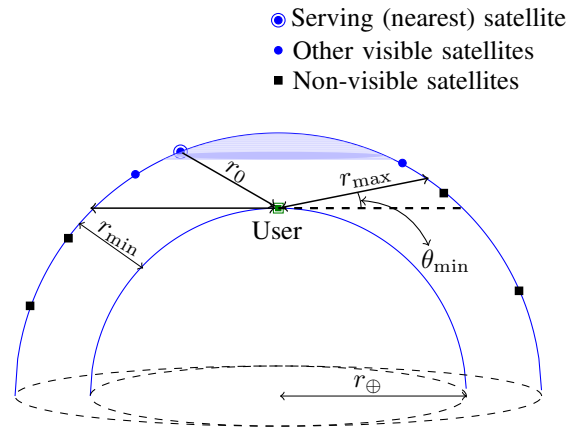

\section{Performance Analysis}

\label{sec:performance}


In order to contribute expressions for coverage probability and average achievable rate of the satellite constellation described in Section~\ref{sec:sysmod}, first, we assume that $N$ satellites are distributed uniformly on a sphere with radius $\rEarth+\rmin$. We will shortly compensate for the performance mismatch generated by the distribution difference between the uniform model and the practical constellations. 

First, we need to characterize some basic distance distributions that stem from the geometry of the considered system. In particular, we express the necessary cumulative distribution function (CDF) and probability density functions (PDFs) in the following lemmas.

\begin{lem}
\label{lem:pdf_r0}
The PDF of the serving distance $R_0$ is given by 
\begin{align}
\label{eq:pdf_r0}
f_{R_0}\left(r_0\right) = N\left(1-\frac{r_0^2-\rmin^2}{4\rEarth(\rEarth+\rmin)}\right)^{N-1} \frac{r_0}{2\rEarth(\rEarth+\rmin)}
\end{align}
for $r_0 \in [\rmin, 2\rEarth+\rmin]$ while $f_{R_0}\left(r_0\right)=0$ otherwise.
\end{lem}
\begin{proof}
We first need to derive the CDF of 
the distance $R$ from any specific one of the satellites in the constellation to the user. From basic geometry, the CDF of the surface area of the shaded spherical cap $A_{\text{cap}}$, formed by any satellite at distance $R$ from the user, in Fig.~\ref{fig:system_model} is $F_{A_{\text{cap}}}\left(a_{\text{cap}}\right)=\frac{a_{\text{cap}}}{4\pi (\rEarth+\rmin)^2}$. Finding a relationship between $A_{\text{cap}}$ and $R$, gives the distribution as
\begin{equation}
\label{eq:cdf_r}
F_{R}\left(r\right) 
= \left\{
\begin{array}{ll}
0, & r<\rmin,\\
\frac{r^2-\rmin^2}{4\rEarth(\rEarth+\rmin)}, &\rmin\leq r \leq2\rEarth+\rmin,\\
1,&r>2\rEarth+\rmin,\\
\end{array} \right.
\end{equation}
and the corresponding PDF is given by  
\begin{equation}
\label{eq:pdf_r}
f_{R}\left(r\right) = \frac{r}{2\rEarth(\rEarth+\rmin)}
\end{equation}
for $r_0 \in [\rmin, 2\rEarth+\rmin]$ while $f_{R}\left(r\right)=0$ otherwise. Due to the channel assignment by which the serving satellite is the nearest one among all the $N$ i.i.d.\ satellites, the CDF of $R_0$ can be expressed as $F_{R_0}\left(r_\text{0}\right)\triangleq\Prob\left(R_0 \leq r_\text{0}\right) = 1 - \left(1-F_{R}\left(r_\text{0}\right)\right)^N$ and, by differentiation, its PDF is $f_{R_0}\left(r_\text{0}\right)= N \left(1-F_{R}\left(r_\text{0}\right)\right)^{N-1}f_r(r)$ which will result in Lemma~\ref{lem:pdf_r0} by substitution from (\ref{eq:cdf_r}) and \eqref{eq:pdf_r}.
\end{proof}

\subsection{Coverage Probability}

In this subsection, we derive the coverage probability of the LEO satellite network for a user in an arbitrary location on Earth. The performance measure of coverage probability is defined as the probability of having 
at least minimum SNR required for successful data transmission. In other words, whenever the SNR of the considered user from its nearest satellite is above the threshold level $T>0$, it is considered to be within the coverage of the satellite communication network.



\begin{proposition}
\label{prop:coverage_prob}
The probability of network coverage for an arbitrarily located user under general fading is
\begin{align}
\label{eq:13}
\nonumber
\Pc\left(T\right)
&\triangleq \Prob\left(\SNR > T\right)\\
\nonumber
&= \frac{N}{2\rEarth(\rEarth+\rmin)}\int_{\rmin}^{\rmax}\left(1-F_{G_0}\left(\frac{Tr_0^{\alpha}\sigma^2}{\ps}\right)\right)\\
&\phantom{=}\times \left(1-\frac{r_0^2-\rmin^2}{4\rEarth(\rEarth+\rmin)}\right)^{N-1}r_0\,dr_0,
\end{align}
where $F_{G_0}(\cdot)$ is the CDF of the channel gain $G_0$.


\end{proposition}
\begin{proof}
To obtain \eqref{eq:13}, we start with the definition of coverage probability:
\begin{align}
\label{eq:16}
\nonumber
&\Pc\left(T\right)=
\mathbb{E}_{R_0}\left[\Prob\left(\SNR>T |R_0\right)\right]\\\nonumber
&=
\int_{\rmin}^{\rmax}\Prob\left(\SNR>T |R_0=r_0\right)f_{R_0}\left(r_0\right)dr_0\\\nonumber
&=\frac{N}{2\rEarth(\rEarth+\rmin)}\int_{\rmin}^{\rmax}\Prob\left(G_0>\frac{Tr_0^{\alpha}\sigma^2}{\ps} \right)\\
&\hspace{80pt}\times\left(1-\frac{r_0^2-\rmin^2}{4\rEarth(\rEarth+\rmin)}\right)^{N-1}r_0\,dr_0.
\end{align}
The upper limit for the integral is due to the fact that the satellites with smaller than $\thmin$ elevation angle have no visibility to the user.  
\end{proof}

The channel characteristics have no effect on the maximum achievable coverage as it is affected only by the geometry of the system model. The following corollary provides the upper bound for coverage probability using Proposition~\ref{prop:coverage_prob}.

\begin{cor}
\label{cor:upperbound}
Setting $T=0$, Proposition~\ref{prop:coverage_prob} leads to an upper bound for coverage probability as  
\begin{align}
\label{eq:28}
\Pc\left(T\right) &\leq F_{R_0}\left(\rmax\right)-F_{R_0}\left(\rmin\right)=1-\left(1-\PV\right)^{{N}},
\end{align}
where $\PV$ is the visibility probability of satellites to the user and is expressed as 
\begin{align}
\label{eq:pv}
\PV = \frac{\rmin-\rmax\sin(\thmin)}{2(\rEarth+\rmin)}.
\end{align}
The expression in \eqref{eq:pv} can be directly obtained as the surface area of the spherical cap, where visible satellites can reside, to the total surface area of the satellites' sphere since the satellites are uniformly distributed.
\end{cor}

Since the number of visible satellites is a binomial random variable with success probability $\PV$, the coverage probability is upper bounded by the probability of observing at least one satellite by the user.


\subsection{Average Data Rate}

In this subsection, we focus on the average achievable data rate. The average achievable rate (in bit/s/Hz) 
states the ergodic capacity from the Shannon--Hartley theorem over a fading communication link normalized to the bandwidth of $1$~Hz.
We can calculate the expression for the average rate of an arbitrary user over generalized fading channels as follows. It is worth noting that the average is taken over both serving distance and fading distributions.

\begin{proposition}
\label{prop:data_rate}
The average rate (in bits/s/Hz) of an arbitrarily located user and its serving satellite under general fading assumption is
\begin{align}
\nonumber
\bar{C} &\triangleq\mathbb{E}\left[\log_2\left(1+\SNR\right)\right]=\frac{N}{2 \ln(2)\rEarth(\rEarth+\rmin)}\\\nonumber
&\times\int_{\rmin}^{\rmax}\int_{0}^{\infty}\ln\left(1+\frac{\ps g_0r_0^{-\alpha}}{\sigma^2}\right)f_{G_0}(g_0)\\
&\hspace{45pt}\times\left(1-\frac{r_0^2-\rmin^2}{4\rEarth(\rEarth+\rmin)}\right)^{N-1}r_0\,dg_0\,dr_0,
\end{align}
where $f_{G_0}(g_0)$ represents the PDF of channel gain $G_0$.
\end{proposition}
\begin{proof}
Taking the expectation over serving distance and channel gain, we have
\begin{align}
\label{eq:24}
\nonumber
&\bar{C}=\mathbb{E}_{G_0,R_0}\left[\log_2\left(1+\SNR\right)\right]\\\nonumber
&=c_0\int_{\rmin}^{\rmax}\mathbb{E}\left[\ln\left(1+\frac{\ps G_0r_0^{-\alpha}}{\sigma^2}\right)\right]\\
&\hspace{45pt}\times\left(1-\frac{r_0^2-\rmin^2}{4\rEarth(\rEarth+\rmin)}\right)^{N-1}r_0\,dr_0,
\end{align}
where $c_0=\frac{N}{2 \ln(2)\rEarth(\rEarth+\rmin)}$.
\end{proof}

\subsection{Effective Number of Satellites}

Due to the fact that satellites in practical constellations are distributed unevenly along different latitudes, i.e., the number of satellites is effectively larger on the inclination limit of the constellation than on equatorial regions, the density of practical deterministic constellations is typically not uniform. Thus, we define and calculate a new parameter, the effective number of satellites, $\Neff$, for every satellite latitude in order to compensate for the uneven density w.r.t.\ practical inclined constellations and create a tight match between the results generated by uniform modeling and those from practical constellation simulations.


\begin{proposition}
\label{prop:Neff}
Let the effective number of satellites ($\Neff$) be the constellation size that corresponds to a satellite density observed by a user on a specific latitude assuming the same density continues everywhere. The effective number of satellites can then be determined as
\begin{align}
\Neff \triangleq \frac{2\,f_{\Phisat}(\phisat)}{\cos(\phisat)} \cdot \Nact,
\end{align}
where random variable $\Phisat$ denotes the latitude of a satellite and $f_{\Phisat}(\phisat)$ corresponds to its PDF.
\end{proposition}
\begin{proof}
The satellite density observed effectively by a user at any latitude assuming that there are $\Neff$ uniformly distributed satellites in total is 
\begin{align}
\label{eq:delta1}
\deltaeff = \frac{\Neff}{4\pi(\rmin+\rEarth)^2},
\end{align}
where the denominator represents the surface area of the satellites' orbital shell. On the other hand, the actual density of the satellites on a ring surface element at latitude $\phisat$ can be written as
\begin{align}
\label{eq:delta2}
\deltaact = \frac{\Nact f_{\Phisat}(\phisat)\,d\phisat}{2\pi(\rmin+\rEarth)^2\cos(\phisat)\,d\phisat},
\end{align}
where the nominator and denominator represent the number of satellites resided in the surface element and the element's surface area, respectively. Setting $\deltaeff = \deltaact$ and applying some simplifications completes the proof.
\end{proof}

\begin{lem}
\label{lem:pdf_phi}
When the satellites' argument of latitude $U$ is a uniform random variable \cite{53}, i.e., $U\sim\mathcal{U}(-\frac{\pi}{2},\frac{\pi}{2})$, the PDF of satellites' latitude with inclination $\iota$ is given by
\begin{align}
f_{\Phisat}(\phisat) = \frac{\sqrt{2}}{\pi} \cdot \frac{\cos(\phisat)}{\sqrt{\cos(2\phisat)-\cos(2\iota)}}
\end{align}
for $\phisat\in[-\iota,\iota]$ while $f_{\Phisat}(\phisat)=0$ otherwise. 
\end{lem}
\begin{proof}
Since the distribution of the argument of latitude is known, we need to find the satellite's latitude as a function of the argument of latitude and inclination angle. The satellite's coordinates can be obtained by multiplication of $\iota$-degree rotation matrix and satellites' orbital plane:
\begin{align}
\nonumber
\begin{bmatrix}
x_s\\
y_s \\
z_s
\end{bmatrix}
&=
\begin{bmatrix}
\cos{(\iota)}& 0 & \sin{(\iota)}\\
0 & 1 & 0\\
-\,\sin{(\iota)}& 0 & \cos{(\iota)}
\end{bmatrix}
    \begin{bmatrix}
(\rmin+\rEarth)\cos{(U)}\\
(\rmin+\rEarth)\sin{(U)} \\
0
\end{bmatrix}\\
&= 
\begin{bmatrix}
(\rmin+\rEarth)\cos{(U)}\cos{(\iota)}\\
(\rmin+\rEarth)\sin{(U)} \\
-\,(\rmin+\rEarth)\cos{(U)}\sin{(\iota)}
\end{bmatrix}.
\label{eq:rotation_matrix}
\end{align}
Therefore, the latitude of the satellite is given as
\begin{align}
    \label{eq:inclination_func}
    \Phisat = g(U) &=\tan^{-1}\left(\frac{z_s}{\sqrt{x_s^2+y_s^2}}\right)\\
    \nonumber
    &=\tan^{-1}\left(-\,\frac{\cos{(U)}\sin{(\iota)}}{\sqrt{\cos^2{(U)}\cos^2{(\iota)}+\sin^2{(U)}}}\right).
\end{align}
The PDF of $\Phisat$ can be written as 
\begin{align}
    \label{eq:PDF_phi}
    f_{\Phisat}(\phisat) = f_{U}(g^{-1}(\phisat))\frac{d}{d\phisat}\left(g^{-1}(\phisat)\right)
\end{align}
by the transform of random variables.
\end{proof}

Thus, when the satellites' argument of latitude is uniform and their inclination is $\iota$, the effective number of satellites can be obtained by using Lemma~\ref{lem:pdf_phi} in Proposition~\ref{prop:Neff} as follows:
\begin{align}
\Neff = \frac{2\sqrt{2}}{\pi}\cdot\frac{1}{\sqrt{\cos(2\phisat)-\cos(2\iota)}} \cdot \Nact.
\end{align}
With high orbit inclination, the effective number of satellites matches to the true number of satellites at some latitudes, decreasing monotonically toward the equator and increasing monotonically toward the poles. By setting $\Nact=\Neff$, we can solve these special latitudes (one for each hemisphere) as
\begin{align}
\phisat = \pm\frac{1}{2} \cos^{-1}\left(\frac{8}{\pi^2} +\cos(2\iota)\right),
\end{align}
if $\iota \geq \frac{1}{2}\cos^{-1}\left(1-\frac{8}{\pi^2}\right) \approx 39.5^{\circ}$ and otherwise $\Neff > \Nact$ for all $\phisat \leq \iota$.


In the special case of having polar orbits (i.e., $\iota=90^\circ$), the PDF of latitude would be the same as the argument of latitude, i.e., $\phisat=U\sim\mathcal{U}(-\frac{\pi}{2},\frac{\pi}{2})$ for all $\phisat$ values. Thus, $\Neff = \frac{2/\pi}{\cos(\phisat)} \cdot \Nact$. For instance, at equator, ($\phisat=0^{\circ}$), $\Neff \approx 0.64\, \Nact$ and by increasing the latitude up to $\phisat \approx  50.5^{\circ}$, we will have $\Neff \approx \Nact$. Finally $\Neff$ will approach to infinity at poles where all satellite orbits cross. The authors up north ($\phiuser=61.5^{\circ}$) at Tampere, Finland observe effectively $30$\% more satellites than there are in reality.

\section{Numerical Results}
\label{sec:Numerical Results}
\begin{table}[!t]
\label{Table: simulation parameters}
\begin{center}
\label{tab:1}
\caption{{Simulation Parameters}}
\begin{tabular}{l|cc}
\hline
\hline
\textbf{Parameters}&\textbf{Values}\\
\hline
Path loss exponent, $\alpha$&2\\
\hline
Rician factor, $K$&100\\
\hline
CDF of channel gains, $F_{G_n}(g_n)$&$1-Q_{1}\left(\sqrt{2K},\sqrt{g_n}\right)$\\
\hline
PDF of channel gains, $f_{G_n}(g_n)$ &$\frac{1}{2}e^{-\frac{g_n+2K}{2}}I_0\left(\sqrt{2Kg_n}\right)$\\
\hline
Transmission power, $\ps$ (W)& 10\\
\hline
Noise power, $\sigma^2$ (dBm)&-93 \\
\hline
 User's latitude, $\phiu$ (degree)&0\\
\hline
\hline
\end{tabular}
\vspace{-0.5 cm }
\end{center}
\end{table}

 The propagation model takes into account the large-scale attenuation with path loss exponent $\alpha=2$, as well as the small-scale fading. To take into account a wider range of fading environments, the channels are assumed to follow Rician fading with parameter $K=100$, where $K$ is the ratio between the direct path received power and other, scattered, paths. The parameter $K$ can be determined according to the type of constellation, i.e., higher $K$ values is suitable when the serving satellite is likely high above the user. As a result, the corresponding channel gains, $G_n$, (being the square of the Rice random variable) have a noncentral chi-squared distribution, $\mathcal{X}^2$, with two degrees of freedom and non-centrality parameter $2K$. Therefore, the CDF and PDF in Propositions~\ref{prop:coverage_prob} and \ref{prop:data_rate} are
\begin{align}
F_{G_0}(g_0) &= 1-Q_{1}\left(\sqrt{2K},\sqrt{g_0}\right),\\
f_{G_0}(g_0) &= \frac{1}{2}e^{-\frac{g_0+2K}{2}}I_0\left(\sqrt{2Kg_0}\right),
\end{align}
respectively, where $Q_{1}(\cdot,\cdot)$ denotes the Marcum Q-function and $I_0(\cdot)$ is the modified Bessel function of the first kind. For producing the numerical results, the transmitted and noise power are set to $\ps=10$~W and $\sigma^2=-93$~dBm, respectively. {The simulation parameters are summarized in Table~I.}

\begin{figure}[t]
    \includegraphics[trim = 8mm 0mm 2mm 0mm, clip,width=0.5\textwidth]{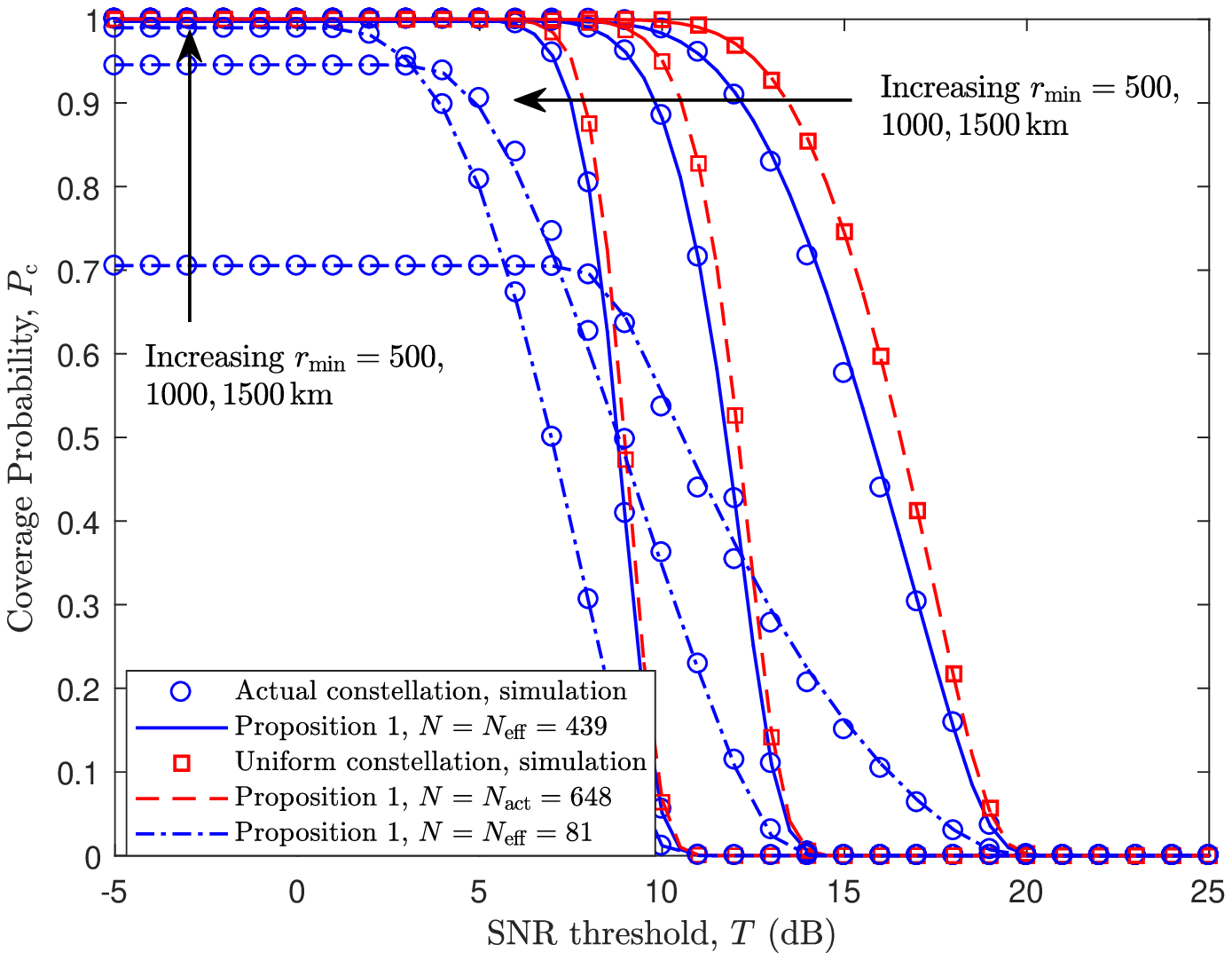}
    \caption{Verification of Proposition~1 with simulations when $K=100$, $\phiuser=0^\circ$, $\iota=70^\circ$, $\rmin\in\{500, 1000, 1500\}$~km, and $\thmin=10^\circ$.}
    \label{fig:CoverageVsT}
    \vspace{-0.5 cm }
\end{figure}

{For numerical verification, we compute the coverage probability and data rate of an actual LEO constellation through Monte Carlo simulations in Matlab to compare them with analytical results presented in this paper.} Figure~\ref{fig:CoverageVsT} verifies the coverage expression given in Proposition~\ref{prop:coverage_prob}, considering different altitudes for $\Nact=648$ and 120. As shown in this figure, there is a slight deviation between the actual constellation described in Section~\ref{sec:sysmod} and uniform constellation performance due to the non-uniform distribution of satellites along different latitudes in the real constellation. Substituting $N=\Neff=439$ and 81 in Proposition~\ref{prop:coverage_prob} which corresponds to $\Nact=648$ and 120, respectively, we can eliminate the mismatch in the coverage. 

For a fewer number of satellites, e.g., $\Nact=120$, it can be well observed from Fig.~\ref{fig:CoverageVsT} that the upper bound for coverage probability, given in Corollary~1, is limited by the probability of observing at least one satellite above the sky. {As a result, the upper bound is enhanced with rising the altitude due to the increase in the visibility probability given in \eqref{eq:pv}. On the other hand, for larger number of satellites, e.g., $\Nact=648$, the performance is affected only by the path loss since the visibility probability approaches one.}
Verification of data rate in Proposition~\ref{prop:data_rate} is shown in Fig.~\ref{fig:datarateVsPt} for different minimum elevation angles. The same as for Fig.~\ref{fig:CoverageVsT}, the mismatch between uniform and actual constellation is omitted by setting $N=\Neff=439$.

Coverage probability versus the total number of satellites for different inclination and minimum required elevation angles is depicted in Fig.~\ref{fig:CoverageVsN}. For plotting this figure, we applied $N=\Neff$ in Proposition~\ref{prop:coverage_prob} in order to compensate for the uneven distribution of satellites along different latitudes. The coverage probability declines with $\thmin$ as the visibility to the user decreases. However, this effect becomes less dominant as the number of satellites increases since the serving satellite, most probably, will be located above the user. Moreover, within the depicted range, the smaller inclination angles result in superior performance due to the larger density of satellites and, consequently, the existence of a stronger serving channel.

There is an optimum altitude for every constellation, as shown in Fig.~\ref{fig:CoverageVsH}, which results in maximum coverage probability. The optimum point increases with rising the minimum elevation angle while the maximum achieved coverage decreases accordingly. The initial increase in the plot is due to the enhancement in the line-of-sight probability of the serving satellite while it is followed by a decline caused by more severe path loss in higher altitudes.

Above results are repeated in terms of data rate in Figs.~\ref{fig:datarateVsN} and~\ref{fig:datarateVsH} w.r.t.\ the total number of satellites and satellite altitude, respectively, using Proposition~\ref{prop:data_rate} with $N=\Neff$. The same as for Fig.~\ref{fig:CoverageVsN}, lower inclination will result in higher data rates in Fig.~\ref{fig:datarateVsN}. However, the impact of both inclination and minimum elevation angle on data rate reduces with increasing the total number of satellites. The same observations as for Fig.~\ref{fig:CoverageVsH} can be also seen in Fig.~\ref{fig:datarateVsH}, except for the optimum altitude differs for maximum coverage probability and data rate.


\begin{figure}[t]
    \includegraphics[trim = 8mm 0mm 2mm 0mm, clip,width=0.5\textwidth]{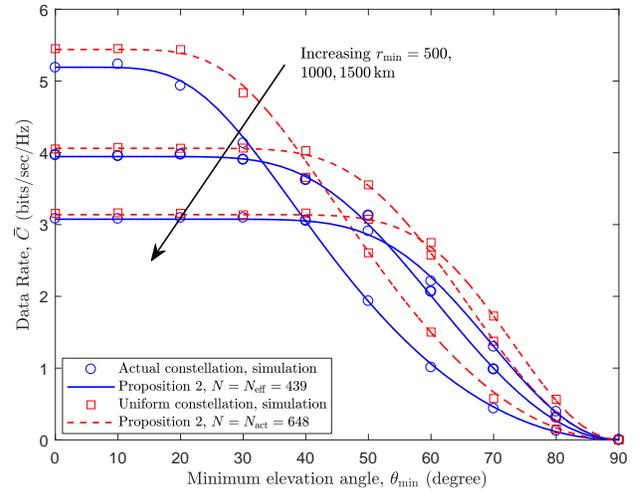}
    \caption{Verification of Proposition~2 with simulations when $K=100$, $\phiuser=0^\circ$, $\iota=70^\circ$, and $\rmin\in\{500, 1000, 1500\}$~km.}
    \label{fig:datarateVsPt}
\end{figure}

\begin{figure}[!t]
    \includegraphics[trim = 8mm 0mm 2mm 7mm, clip,width=0.5\textwidth]{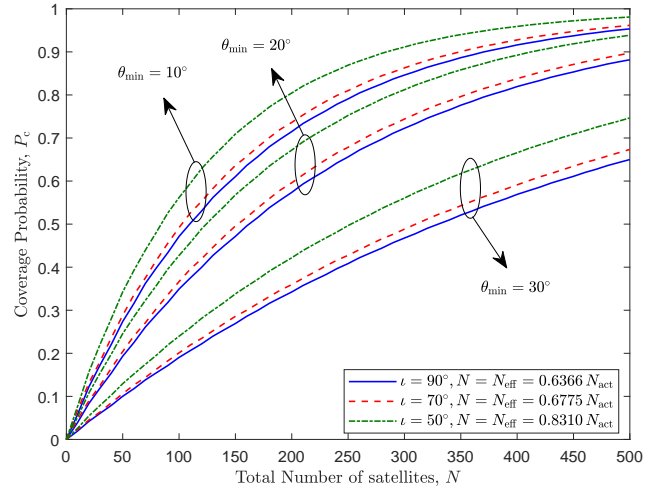}
    \caption{Coverage probability for different constellation sizes when $K=100$, $\phiuser=0^\circ$, $T=10$~dB, and $\rmin=500$~km.}
    \label{fig:CoverageVsN}
\end{figure}

\section{Conclusions}
\label{sec:Conclusion}
In this paper, we presented a tractable approach for uplink and downlink coverage and rate analysis of low Earth orbit satellite networks. The satellite network is, first, modeled with a uniform distribution which was then applied to obtain exact expressions for coverage probability and data rate of an arbitrary user in terms of network parameters. The slight deviation between the performance metrics of the uniform and actual constellations was compensated by derivation of a new parameter---effective number of satellites---to take into account the effect of uneven satellite distribution along different latitudes. The proposed framework in this paper paves the way for accurate analysis, optimization and design of the future dense satellite networks. 

\vfill





\begin{figure}[!t]
    \includegraphics[trim = 8mm 0mm 2mm 0mm, clip,width=0.5\textwidth]{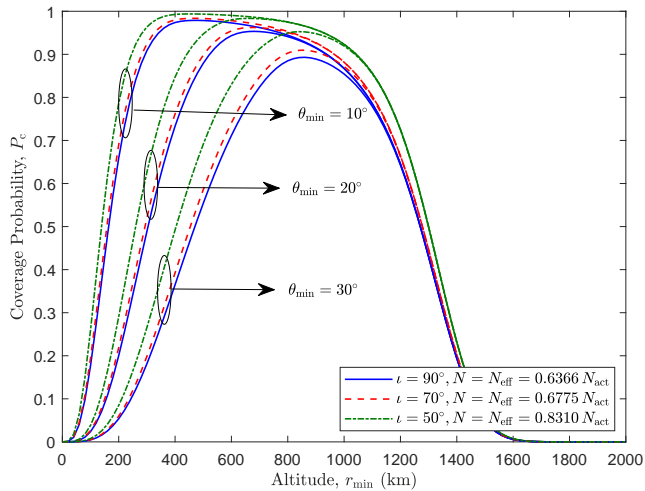}
    \caption{Coverage probability for different altitudes when $K=100$, $\phiuser=0^\circ$, $T=10$~dB, and $\Nact=648$.}
    \label{fig:CoverageVsH}
\end{figure}

\begin{figure}[!t]
    \includegraphics[trim = 8mm 0mm 2mm 0mm, clip,width=0.5\textwidth]{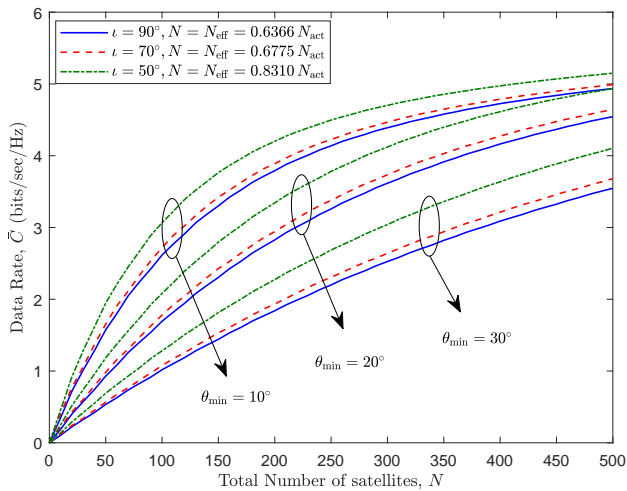}
    \caption{Data rate for different constellation sizes when $K=100$, $\phiuser=0^\circ$, and $\rmin=500$~km.}
    \label{fig:datarateVsN}
\end{figure}

\begin{figure}[!t]
    \includegraphics[trim = 8mm 0mm 2mm 0mm, clip,width=0.5\textwidth]{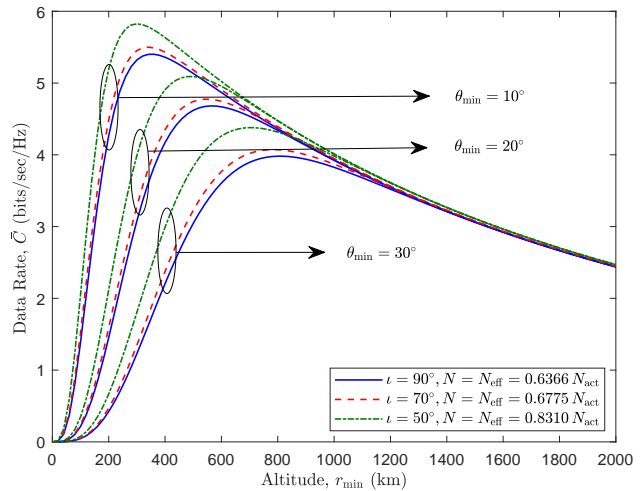}
    \caption{Data rate for different altitudes when $K=100$, $\phiuser=0^\circ$, and $\Nact=648$.}
    \label{fig:datarateVsH}
\end{figure}

\bibliography{ref} 
\bibliographystyle{IEEEtran}

\end{document}

%% file: system_model.tex
\begin{tikzpicture}
\draw [blue, domain=-3.5: 3.5, samples=200] plot(\x, {sqrt(3.5^2-\x*\x)});
\draw [blue, domain=-2.5: 2.5, samples=200] plot(\x, {sqrt(2.5^2-\x*\x)});
\draw[black,dashed](0,0) ellipse (2.5cm and 0.4cm) node (E) [inner sep = 0,blue] {};
\draw[black,dashed](0,0) ellipse (3.5cm and 0.4cm) node (S) [inner sep = 0,blue] {};

\draw[blue,fill](-1.3,3.25) circle[radius=0.05] node (Serv) [inner sep = 0,black] {};
\draw[blue](-1.3,3.25) circle[radius=0.09] node (Serv2) [black] {};

\draw[blue,fill](1.65,3.1) circle[radius=0.05] [inner sep = 0,blue] {};
\draw[blue,fill](-1.9,2.95) circle[radius=0.05] [inner sep = 0,blue] {};
 \node [inner sep = 0,draw, fill, shape=rectangle, minimum width=0.1cm, minimum height=0.1cm,  anchor=center] at (2.2,2.7) {};

 \node [inner sep = 0,draw, fill, shape=rectangle, minimum width=0.1cm, minimum height=0.1cm, anchor=center] at (3.2,1.4) {};
 \node [inner sep = 0,draw, fill, shape=rectangle, minimum width=0.1cm, minimum height=0.1cm,  anchor=center] at (-2.79,2.1) {};
 \node [inner sep = 0,draw, fill, shape=rectangle, minimum width=0.1cm, minimum height=0.1cm,  anchor=center] at (-3.3,1.2) {};
 \def\r{3.5}
    \def\H{2.5}
  \begin{scope}
    \clip 
    [domain=-3.5: 3.5] plot(\x, {sqrt(3.5^2-\x*\x)});
   \shade[top color=blue!50!white,bottom color=blue!20!white,opacity=0.4] 
 ({-\r},{0.9*\r}) rectangle ++({1.5*\r},{0.1*\r+\H});
    \fill[blue!25] (0,3.2) circle [x radius={1.4},
   y radius={0.05}];
    \end{scope}

 \node[green!50!black,fill] (user) [inner sep = 0,draw, shape=rectangle, minimum width=0.07cm, minimum height=0.07cm, anchor=center] at (0,2.5) {};
\node[green!50!black]  [inner sep = 0,draw, shape=rectangle, minimum width=0.13cm, minimum height=0.13cm, anchor=center] at (0,2.5) {};
\node[text width=2.5cm] at (0.9,2.2)     {User};

\draw[blue,fill](0,5) circle[radius=0.05] ;
\draw[blue](0,5) circle[radius=0.09] ;
\node[] at (2,5) {Serving (nearest) satellite};

 \node [inner sep = 0,draw, fill, shape=rectangle, minimum width=0.1cm, minimum height=0.1cm, anchor=center] at (0,4.2) {};
\node[] at (1.7,4.2) {Non-visible satellites};

\draw[blue,fill](0,4.6) circle[radius=0.05] ;
\node[] at (1.8,4.6) {Other visible satellites};

\draw[line width=0.3 mm, black,dashed] (0,2.5) -- (2.5,2.5);
\draw[inner sep = 0,line width=0.2 mm, black,->] (Serv) -- (user);
\draw[line width=0.1 mm, line width=0.2 mm,black,<->] (user) -- (2.0,2.9) ;
\draw[line width=0.1 mm, black,<->] (-2.65,2.3) -- (-1.8,1.7);
\node[text width=3cm,rotate=-30] at (-1.1,1.5)    {${\rmin}$};
\node[text width=3cm,rotate=-30] at (0.6,2.3)    {${r_0}$};
\node[text width=3cm,rotate=10] at (2.3,3.1)    {${\rmax}$};
\draw (1.1,2.5) arc (0:25:0.5) node (arcc) [inner sep = 0]{};
\draw[inner sep = 0,line width=0.2 mm, black,<->] (user) -- (-2.5,2.5);
\draw[inner sep = 0, line width=0.1 mm, black,<->] (E) -- (2.5,0);
\node[text width=3cm] at (2.5,0.15)     {\textit{$\rEarth$}};
\draw[bend left,<->]  (1.1,2.6) to (2,2);
\node[text width=3cm] at (3.4,1.8)     {$\thmin$};
 ambiguity
\end{tikzpicture}
\par